\documentclass[12pt]{article}

\usepackage{amsmath}
\usepackage{amssymb}
\usepackage{appendix}
\usepackage{amsthm}
\usepackage{enumitem}
\usepackage{hyperref}

\usepackage[margin=0.8in]{geometry}

\numberwithin{equation}{section}

\newtheorem{thm}{Theorem}[section]
\newtheorem{cor}[thm]{Corollary}
\newtheorem{defi}{Definition}[section]
\newtheorem{assume}{Assumption}
\newtheorem{prop}{Propositioin}[section]

\newtheorem{lemma}{Lemma}[section]

\begin{document}

\author{Jihun Han\thanks{jihunhan@cims.nyu.edu}\; and Hyungbin Park\thanks{hyungbin@cims.nyu.edu, hyungbin2015@gmail.com}\\ \\
Courant Institute of Mathematical Sciences,\\ New York University, New York, USA \\ \\
{\normalsize First version: Nov 09, 2014} \\ {\normalsize Final version: Mar 03, 2015}}
\title{The Intrinsic Bounds on the Risk Premium of Markovian Pricing Kernels}

\date{}

\maketitle

\begin{abstract}
The risk premium is one of main concepts in mathematical finance.
It is a measure of the trade-offs investors make between return and risk and is defined by
the excess return relative to the risk-free interest rate that is earned from an asset per one unit 
of risk.
The purpose of this article is to determine upper and lower bounds on the risk premium of an 
asset based on the market prices of options.
One of the key assumptions to achieve this goal is that the market is Markovian.  
Under this assumption, we can transform the problem of finding the bounds into a second-order 
differential equation. We then obtain upper and lower bounds on the risk premium by analyzing 
the differential equation.
\end{abstract}

\section{Introduction}
\label{sec:intro}
The {\em risk premium} or {\em market price of risk} is one of main concepts in mathematical 
finance.
The risk premium is a measure of the trade-offs investors make between return and risk and is 
defined by
the excess return relative to the risk-free interest rate earned from an asset per one unit of risk.
The risk premium determines the relation between an objective measure and a risk-neutral 
measure.
An objective measure describes the actual stochastic dynamics of markets, and a risk-neutral 
measure determines the prices of options.

Recently, many authors have suggested that the risk premium (or, equivalently, objective 
measure) can be determined from a risk-neutral measure.
Ross \cite{Ross13} demonstrated that the risk premium can be uniquely determined by a risk-neutral measure. His model assumes that there is a finite-state Markov process $X_{t}$ that 
drives the economy in discrete time $t\in\mathbb{N}.$
Many authors have extended his model to a continuous-time setting using a Markov diffusion 
process
$X_{t}$ with state space $\mathbb{R}$; see, e.g., 
\cite{Borovicka14},\cite{Carr12},\cite{Dubynskiy13},\cite{Goodman14},\cite{Park14b},\cite{Qin14b} and 
\cite{Walden13}.
Unfortunately, in the continuous-time model, the risk premium is not uniquely determined from a 
risk-neutral measure \cite{Goodman14}, \cite{Park14b}.

To determine the risk premium uniquely, all of the aforementioned authors assumed that some 
information about the objective measure was known or restricted the process $X_t$ to some 
class.   
Borovicka, Hansen and Scheinkman \cite{Borovicka14} made the assumption that the process 
$X_t$ is {\em stochastically stable} under the objective measure.
In \cite{Carr12}, Carr and Yu assumed that the process $X_{t}$ is a {\em bounded} process.
Dubynskiy and Goldstein \cite{Dubynskiy13} explored Markov diffusion models with {\em 
reflecting boundary} conditions.
In \cite{Park14b}, Park assumed that $X_t$ is non-attracted to the left (or right) boundary under 
the objective measure.
Qin and Linetsky \cite{Qin14b} and Walden \cite{Walden13} assumed that the process $X_{t}$ is 
{\em recurrent} under the objective measure.
Without these assumptions, one cannot determine the risk premium uniquely.

The purpose of this article is to investigate the bounds of the risk premium.
As mentioned above, without further assumptions, the risk premium is not uniquely determined, 
but one can determine upper and lower bounds on the risk premium.
To determine these bounds, we need to consider how the risk premium of an asset is 
determined in a financial market.

A key assumption of this article is that the reciprocal of the pricing kernel
is expressed in the form
$e^{\beta t}\,\phi(X_{t})$
for some positive constant $\beta$ and positive function $\phi(\cdot).$ 
For example, in the {\em consumption-based capital asset model} \cite{Campbell99}, \cite{Karatzas98}, the pricing kernel is expressed in the above form.
We will see that in this case
the risk premium $\theta_t$ is given by
\begin{equation} \label{eqn:theta}
\theta_t=(\sigma\phi'\phi^{-1})(X_t)\;,
\end{equation}
where $\sigma(X_t)$ is the volatility of $X_t.$

The problem of determining the bounds of the risk premium can be transformed into a second-order differential equation. We will demonstrate that $\phi(\cdot)$ satisfies the following 
differential equation:
$$\mathcal{L}\phi(x):=\frac{1}{2}\sigma^2(x){\phi ''(x)}+k(x)\phi '(x)
-r(x)\phi (x) =-\beta \,\phi (x)$$
for some unknown positive number $\beta.$
Thus, we can determine the bounds of the risk premium by investigating the bounds of $(\sigma\phi'\phi^{-1})(\cdot)$ for a 
positive solution $\phi(\cdot).$
It will be demonstrated that two special solutions of $\mathcal{L}h=0$ play an important role for 
the bounds of the risk premium $\theta_t.$

The following provides an overview of this article.
In Section \ref{sec:Markovian_pricing}, we state the notion of Markovian pricing kernels.
In Section \ref{sec:risk_premium}, we investigate the risk premium of an asset
and see how the problem of determining the bounds of the risk premium is transformed into a 
second-order differential equation.
In Section \ref{sec:intrinsic_bounds}, we find upper and lower bounds on the risk premium of an asset, which is the main result of this article.
In Section \ref{sec:appli}, we see how this result can be applied to determine the range of return of an asset.
Finally, Section \ref{sec:conclusion} summarizes this article.

\section{Markovian pricing kernels}
\label{sec:Markovian_pricing}

A financial market is defined as a probability space
$(\Omega,\mathcal{F},\mathbb{P})$ having a Brownian motion $B_{t}$
with the filtration $\mathcal{F}=(\mathcal{F}_{t})_{t=0}^{\infty}$ generated by $B_{t}$.  
All the processes
in this article are assumed to be adapted to the filtration $\mathcal{F}$. 
$\mathbb{P}$ is the objective measure of this market.

\begin{assume}
In the financial market, there are two assets. One is
a {\em money market account} $e^{\int_0^t\,r_s\,ds}$ 
with an {\em interest rate} process $r_{t}$
and the other is a risky asset $S_{t}$ satisfying 
$$dS_t=\mu_tS_t\,dt+v_tS_t\,dB_t\;.$$
\end{assume}

\noindent Throughout this article, the stochastic discount factor is the money market account.

Let $\mathbb{Q}$ be a risk-neutral measure in the market $(\Omega,\mathcal{F},\mathbb{P})$
such that 
$S_t\,e^{-\int_0^tr_s\,ds}$ is a local martingale under $\mathbb{Q}.$
Put the Radon-Nikodym derivative 
\begin{equation*} 
\Sigma_{t}=\left.\frac{d \mathbb{Q}}{d \mathbb{P}}
\right|_{\mathcal{F}_{t}} \; , 
\end{equation*}
which is known to be a martingale process 
on $(\Omega,\mathcal{F},\mathbb{P}).$ 
We can write in the SDE form 
\begin{equation*} 
d\Sigma_{t}=-\theta_{t}\Sigma_{t}\, dB_{t}
\end{equation*}
where
\begin{equation}\label{eqn:rho}
\theta_{t}:=\frac{\mu_t-r_t}{v_t}
\end{equation}
is the {\em risk-premium} or {\em market price of risk.}
 It is well-known that $W_{t}$ defined by
\begin{equation}\label{eqn:Girsanov}
dW_{t}=\theta_{t}dt+dB_{t}
\end{equation}
 is a Brownian motion under $\mathbb{Q}.$
We define the reciprocal of the {\em pricing kernel} by $L_{t}=e^{\int_0^tr_s\,ds}/\Sigma_{t}.$  Using the Ito formula, 
\begin{equation} \label{eqn:RN_SDE}
\begin{aligned}
dL_t&=(r_{t}+\theta_{t}^2)\,L_t\,dt +\theta_{t}L_t\, dB_{t}\\
&=r_{t}L_t\,dt + \theta_{t}L_t\, dW_{t} 
\end{aligned}
\end{equation}
is obtained.

\begin{assume} \label{assume:Markovian}
Assume that (the reciprocal of) the pricing kernel $L_t$ is Markovian in the sense that 
there are a positive function $\phi\in C^{2}(\mathbb{R}),$ a positive number $\beta$ 
and a state variable $X_t$ 
such 
that
\begin{equation} \label{eqn:Markovian}
L_{t}=e^{\beta t}\,\phi(X_{t})\,\phi^{-1}(X_{0})\; .
\end{equation}
In this case, we say $(\beta,\phi)$ is a
{\em principal pair} of $X_{t}.$
\end{assume}

\noindent We imposed a special structure on the pricing kernel. This specific form is commonly assumed in the recovery literature as in \cite{Borovicka14},\cite{Carr12},\cite{Dubynskiy13},\cite{Goodman14},\cite{Park14b},\cite{Qin14b},\cite{Ross13} and \cite{Walden13}.
In general, $L_t$ can be expressed as 
$$L_{t}=e^{\beta t}\,\phi(X_{t})\,\phi^{-1}(X_{0})\,M_t$$
where $M_t$ is a $\mathbb{Q}$-martingale. Refer to \cite{Hansen09} for this general expression.
Assumption \ref{assume:Markovian} has an implication that the martingale term $M_t$ is equal to $1.$

We now shift our attention to the assumption that $\beta>0.$ 
In lots of literature on asset pricing theory, $\beta$ is the discount rate of the representative agent, which is typically a positive number.
For example, in the {\em consumption-based capital asset model} \cite{Campbell99}, \cite{Karatzas98},
(the reciprocal of) the pricing kernel is expressed by
$$e^{\beta t}\frac{U'(c_0)}{U'(c_t)}$$
where $U$ is the utility of the representative agent, $c_t$ is the aggregate consumption process and $\beta$ is the discount rate of the agent.

\begin{assume} \label{assume:X}
The state variable $X_{t}$ is a
time-homogeneous Markov diffusion process satisfying the following SDE.
$$dX_{t}=k(X_{t})\,dt+\sigma(X_{t})\,dW_{t}\,,\;X_{0}=\xi\;.$$
$k(\cdot)$ and $\sigma(\cdot)$ are assumed to be known ex ante.
The process $X_t$ takes values in some interval $I$ with endpoints $c$ and $d,\,-\infty\leq c<d\leq\infty.$
It is assumed that $b(\cdot)$ and $\sigma(\cdot)$ are continuous on $I$ and continuously differentiable on $(c,d)$
and that $\sigma(x)>0$ for $x\in(c,d).$
\end{assume}

\begin{assume} \label{assume:interest_rate}
The short interest rate $r_t$ is determined by $X_{t}.$ More precisely, 
there is a continuous positive function
$r(\cdot)$ such that $r_{t}=r(X_{t}).$
\end{assume}
\noindent Under these assumptions, the next section demonstrates how to transform the problem of
determining the bounds of the risk premium into a second-order differential equation.
We will also describe the properties of positive solutions of the differential equation.

\section{Risk premium}
\label{sec:risk_premium}
The purpose of this article is to determine upper and lower bounds on the risk premium 
$\theta_t.$ 
First, we investigate how the risk premium $\theta_t$ is determined with the Markovian pricing kernel.
Applying the Ito formula to \eqref{eqn:Markovian}, we have
$$dL_{t}=\left(\beta+\frac{1}{2}(\sigma^{2}\phi''\phi^{-1})(X_{t})+(k\phi'\phi^{-1})(X_{t})\right)L_{t}
\,dt+(\sigma\phi'\phi^{-1})(X_{t})\,L_{t}\,dW_{t}$$
and by \eqref{eqn:RN_SDE}, we know
$dL_{t}=r(X_{t})\,L_{t}\,dt + \theta_{t}L_{t}\, dW_{t}.
$
By comparing these two equations, we obtain
$$
\frac{1}{2}\sigma^2(x){\phi ''(x)}+k(x)\phi '(x)
-r(x)\phi (x) =-\beta \,\phi (x)
$$
and
\begin{equation}\label{eqn:Markovian_theta}
\theta_{t}=(\sigma\phi'\phi^{-1})(X_{t})\;.
\end{equation}
Define a infinitesimal operator $\mathcal{L}$ by
$$\mathcal{L}\phi(x)=\frac{1}{2}\sigma^2(x){\phi ''(x)}+k(x)\phi '(x)
-r(x)\phi (x)\;.$$

\begin{thm}
Under Assumption 1-\ref{assume:interest_rate}, let $(\beta,\phi)$ be
a principal pair of $X_{t}.$ Then, $(\beta,\phi)$
satisfies 
$\mathcal{L}\phi=-\beta\phi\;.$
\end{thm} \noindent
\noindent We also have the following theorem by \eqref{eqn:Girsanov} and \eqref{eqn:rho}.
\begin{thm} \label{thm:theta_phi}
The risk premium is given by
$\theta_{t}=\theta(X_{t})$
where $\theta(\cdot):=(\sigma\phi'\phi^{-1})(\cdot).$
We thus have that $dB_{t}=-\theta(X_{t})\,dt+dW_{t}.$
\end{thm}
\noindent This theorem explains the relation between the risk premium and the pricing kernel $L_t.$

The purpose of this article is to determine upper and lower bounds on $\theta(\cdot)$ based on $k(\cdot),\,\sigma(\cdot)$ and $r(\cdot).$
The positive function $\phi(\cdot)$ and the positive number $\beta$ are assumed to be 
unknown.
The main idea is to determine the properties of all of the possible $\phi(\cdot)$'s and $\beta$'s and then to 
obtain upper and lower bounds on the possible $(\phi'\phi^{-1})(\cdot)$ values.
From \eqref{eqn:Markovian_theta}, we can determine the bounds of the risk premium $\theta_t.$

\section{Intrinsic bounds}
\label{sec:intrinsic_bounds}
We are interested in a solution pair $(\lambda,h)$ of
$\mathcal{L}h=-\lambda h$ with positive function $h.$
There are two possibilities.
\begin{itemize}[noitemsep,nolistsep]
\item[\textnormal{(i)}] there is no positive solution $h$ for any $\lambda\in\mathbb{R}$, or
\item[\textnormal{(ii)}] there exists a number $\overline{\beta}$ such that it
has two linearly independent positive solutions for $\lambda<\overline{\beta},$ has no positive solution for
$\lambda>\overline{\beta}$ and has one or two linearly independent solutions for $\lambda=\overline{\beta}.$  
\end{itemize}
Refer to page 146 and 149 in \cite{Pinsky}.
In this article, we implicitly assumed the second case by Assumption \ref{assume:Markovian}.
\begin{defi}
For each $\lambda$ with $\lambda\leq\overline{\beta},$ 
we say $(\lambda,h)$ is a candidate pair if $(\lambda,h)$ is a solution pair of $\mathcal{L}h=-\lambda h$ and if $h(\xi)=1$ (i.e., $h$ is normalized).
We define the candidate set by
$$\mathcal{C}_\lambda:=\{\,h'(\xi)\in\mathbb{R}\,|\,\mathcal{L}h=-
\lambda h,\, h(\xi)=1,\,h(\cdot)>0 \,\}\;.$$
\end{defi}
\noindent It is known that $\mathcal{C}_\lambda$ is a connected compact set. Refer to \cite{Park14b} or \cite{Pinsky}.
Denote the functions corresponding to $\max\mathcal{C}_\lambda$ and $\min\mathcal{C}_\lambda$
by $H_\lambda$ and $h_\lambda,$ respectively.
It is assumed that $H_\lambda(\xi)=h_\lambda(\xi)=1.$

For a solution pair $(\lambda,h)$ with $h>0,$
it is easily checked that
$$e^{\lambda t-\int_{0}^{t} r(X_{s})ds}\,h(X_{t})\,h^{-1}(\xi)$$
is a local martingale under $\mathbb{Q}.$
To be a Radon-Nikodym derivative, this should be a martingale. 
Thus, we are interested in solution pairs that induces martingales.
\begin{defi}
Let $(\lambda,h)$ be a candidate pair. 
We say $(\lambda,h)$ is a admissible pair if $$e^{\lambda t-\int_{0}^{t} r(X_{s})ds}\,h(X_{t})\,h^{-1}(\xi)$$ is a martingale under $\mathbb{Q}.$ In this case, 
a measure obtained
from the risk-neutral measure $\mathbb{Q}$
by the Radon-Nikodym derivative 
$$\left.\frac{\,d\,\cdot\,}{d\mathbb{Q}}\right|_{\mathcal{F}_{t}}=e^{\lambda t-\int_{0}^{t} r(X_{s})ds}\,h(X_{t})\,h^{-1}(\xi)$$
is called
{\em the transformed measure} with respect to the pair $(\lambda,h).$
\end{defi}

We now investigate the bounds of the risk premium.
We set
\begin{equation*}
\begin{aligned}
\ell:&=\inf\{0\leq\lambda\leq\overline{\beta}\,|\,(\lambda,h_\lambda)\text{ is an admissible pair}\} \\
L:&=\inf\{0\leq\lambda\leq\overline{\beta}\,|\,(\lambda,H_\lambda)\text{ is an admissible pair}\}\;.
\end{aligned}
\end{equation*}
Two functions $h_\ell$ and $H_L$ will play a crucial role in determining the bounds of the risk 
premium $\theta_t.$
To see this, we need the following proposition.
\begin{prop} \label{prop:relation}
Let $\alpha<\lambda$ and let $(\lambda,h)$ be a candidate. Then,
\begin{equation*}
\begin{aligned}
(h_\alpha'h_\alpha^{-1})(x)\leq(h'h^{-1})(x)\leq(H_\alpha'H_\alpha^{-1})(x)\;.
\end{aligned}
\end{equation*}
\end{prop}
\noindent See Appendix \ref{app:pf_relation} for proof. 
This proposition says that if a candidate pair $(\lambda,h)$ with $\lambda>0$ is an admissible pair, then 
\begin{equation*}
\begin{aligned}
(h_\ell'h_\ell^{-1})(x)\leq(h'h^{-1})(x)\leq(H_L'H_L^{-1})(x)\;.
\end{aligned}
\end{equation*}
This equation gives upper and lower bounds on the risk premium.
The only information that we know about the principal pair $(\beta,\phi)$ is that $\beta>0$ and that
$(\beta,\phi)$ is an admissible pair.
Thus, we can 
conclude that
$$(h_\ell'h_\ell^{-1})(x)\leq(\phi'\phi^{-1})(x)\leq(H_L'H_L^{-1})(x)\;.$$
By Theorem \ref{thm:theta_phi}, we obtain the main theorem of this article.

\begin{thm} \label{thm:intrinsic_bounds}\textnormal{(Intrinsic Bounds of Risk Premium)} \newline
Let $\theta_t$ be the risk premium. Then,
\begin{equation*}
\begin{aligned}
(\sigma h_\ell'h_\ell^{-1})(X_t)\leq\theta_t\leq(\sigma H_L'H_L^{-1})(X_t)\;.
\end{aligned}
\end{equation*}
\end{thm}
\noindent This theorem implies that
we can determine the range of the risk premium when a risk-neutral measure is given.
Upper and lower bounds can then be calculated using option prices.
In the next section, as an example, we show that in the classical Black-Scholes model, the risk premium satisfies
$$-\frac{2r}{v}\leq\theta_t\leq v$$
where $r$ is the interest rate and $v$ is the volatility of the stock.

\begin{cor}
Let $\theta_t$ be the risk premium. Then,
\begin{equation*}
\begin{aligned}
(\sigma h_0'h_0^{-1})(X_t)\leq\theta_t\leq(\sigma H_0'H_0^{-1})(X_t)\;.
\end{aligned}
\end{equation*}
\end{cor}

\noindent This gives rough bounds on the risk premium. In general, two solutions $h_0$ and $H_0$ are more straightforward to find than $h_\ell$ and $H_L.$ \newline

We can find better bounds if some information of $\mathbb{P}$-dynamics of $X_t$ is known.
As mentioned in Section \ref{sec:intro}, if $X_t$ is recurrent under the objective measure, then we can find the exact risk premium. Thus, the bounds is not informative.
For example, if the state variable is an interest rate, which is usually recurrent under the objective measure, then we can find the precise risk premium.
For more details, see \cite{Borovicka14},\cite{Park14b},\cite{Qin14b} and \cite{Walden13}.

We can find better bounds when it is known that the state variable is non-attracted to the left (or right) boundary under the objective measure.
This is a reasonable assumption under some situations. For example, a stock price process is usually not attracted to zero boundary. 
\begin{prop}
The process $X_t$ is non-attracted to the left boundary under the transformed measure with respect to $(\lambda,h)$ if only if
$h=H_\lambda.$
\end{prop}
\noindent See \cite{Park14b} for proof. Similarly, the process $X_t$ is non-attracted to the right boundary under the transformed measure with respect to $(\lambda,h)$ if only if
$h=h_\lambda.$ This proposition says the following theorem.
\begin{thm} \label{thm:non-attracted}
If the state process $X_t$ is non-attracted to the left boundary, then 
the risk premium $\theta_t$ satisfies
$$(\sigma H_{\overline{\beta}}'H_{\overline{\beta}}^{-1})(X_t)\leq\theta_t\leq(\sigma H_L'H_L^{-1})(X_t)\;.$$
\end{thm}
\noindent Similarly, if the state process $X_t$ is non-attracted to the right boundary, then 
the risk premium $\theta_t$ satisfies
$$(\sigma h_\ell'h_\ell^{-1})(X_t)\leq\theta_t\leq(\sigma h_{\overline{\beta}}'h_{\overline{\beta}}^{-1})(X_t)\;.$$

\section{Returns of Stock}
\label{sec:appli}

In this section, we investigate the bounds of the risk premium when the state variable $X_t$ 
is the stock price process $S_t.$
In practice, the S$\&$P 500 index process, which can theoretically be regarded as a stock price 
process, is used as a state variable \cite{Audrino14}.
In this case, we can determine upper and lower bounds on the return of the stock process 
$S_t.$
Suppose that $S_t=X_t$ and that the interest rate is constant $r.$
Under the risk-neutral measure, the dynamics of $X_t$ is
$$dX_t=rX_t\,dt+\sigma(X_t)\,X_t\,dW_t\;.$$  
By Theorem \ref{thm:intrinsic_bounds}, the risk premium satisfies
$$(\sigma h_\ell'h_\ell^{-1})(x)\leq\theta(x)\leq(\sigma H_L'H_L^{-1})(x)$$
where $h_\ell$ and $H_L$ are the corresponding solutions. 
We obtain upper and lower bounds on the return $\mu_t$ by using equation \eqref{eqn:rho}.
$$r+(\sigma^2h_\ell'h_\ell^{-1})(X_t)\leq\mu_t\leq r+(\sigma^2H_L'H_L^{-1})(X_t)\;.$$

As an example, we explore the classical Black-Scholes model for stock price $X_t.$   
$$dX_t=rX_t\,dt+vX_t\,dW_t\,,\;X_0=1$$ 
for $v>0.$ The infinitesimal operator is 
$$\mathcal{L}h(x)=\frac{1}{2}v^2x^2h''(x)+rxh'(x)-rh(x)\;.$$ 
It can be easily shown that every solution pair of $\mathcal{L}h=-\lambda h$ induces a martingale, that is, every candidate pair is an admissible pair. Thus it is obtained that $\ell=L=0.$
We want to find the positive solutions of $\mathcal{L}h=0$
with $h(0)=1.$
The solutions are given by $h(x)=c\,x^{-\frac{2r}{v^2}}+(1-c)x$
for $0\leq c\leq 1.$
Thus $$h_0(x)=x^{-\frac{2r}{v^2}},\; H_0(x)=x\;.$$
The risk premium $\theta_t$ satisfies
$-\frac{2r}{v}\leq\theta_t\leq v.$
The upper and lower bounds of the return $\mu_t$ of $X_t$ is given by
$-r\leq\mu_t\leq r+v^2.$

We can find better bounds if we know that the stock price is non-attracted to zero boundary. 
By direct calculation, we have that $\overline{\beta}=\frac{(r+\frac{1}{2}v^2)^2}{2v^2}$ and that
$$H_{\overline{\beta}}(x)=x^{\frac{1}{2}-\frac{r}{v^2}},\; H_0(x)=e^x\;.$$
By Theorem \ref{thm:non-attracted}, the risk premium $\theta_t$ satisfies
$\frac{1}{2}-\frac{r}{v}\leq\theta_t\leq v.$
The upper and lower bounds of the return $\mu_t$ of $X_t$ is given by
$\frac{v}{2}\leq\mu_t\leq r+v^2.$

\section{Conclusion}
\label{sec:conclusion}
This article determined the possible range of the risk premium of the market using the market prices of options.
One of the key assumptions to achieve this result is that the market is Markovian driven by a state variable $X_t.$  
Under this assumption, we can transform the problem of determining the bounds into a second-order differential equation. We then obtain the upper and lower bounds of the risk premium by 
analyzing the differential equation.

We illuminated how the problem of the risk premium is transformed into a problem described by 
a second-order differential equation. The risk premium is determined by
$\theta_t=(\sigma\phi'\phi^{-1})(X_t)$
with a positive function $\phi(\cdot).$
We demonstrated that $\phi(\cdot)$ satisfies
$\mathcal{L}\phi=-\beta\phi$
for some positive number $\beta$, where $\mathcal{L}$ is a second-order operator that is 
determined by option prices.

We demonstrated that two special solutions $h_\ell$ and $H_L$ play a crucial role for determining upper and lower bounds on the risk premium.
The risk premium $\theta_t$ satisfies
\begin{equation*}
\begin{aligned}
(\sigma h_\ell'h_\ell^{-1})(X_t)\leq\theta_t\leq(\sigma H_L'H_L^{-1})(X_t)\;.
\end{aligned}
\end{equation*}
We also discussed better bounds when it is known that the state variable is non-attracted to the left (or right) boundary.
It is stated how this result can be applied to determine the range of return of an asset.

The following extensions for future research are suggested. First, it would be interesting to 
extend the process $X_t$ to a multidimensional process or a process with jump.
Second, it would be interesting to determine the bounds of the risk premium when the process 
$X_t$ is a non-Markov process. In this article, we discussed only a time-homogeneous Markov 
process.
Third, it would be interesting to explore more general forms of
the pricing kernel. We discussed only the case in which (the reciprocal of) the pricing kernel has 
the
form $e^{\beta t}\phi(X_{t}).$

\section*{Acknowledgement}
The authors highly appreciate the detailed valuable comments of the referee.

\appendix

\section{Proof of Proposition \ref{prop:relation}}
\label{app:pf_relation}

Hereafter, without loss of generality, we assume that $X_0=0,$ the left boundary of the range of $X_t$ is $-\infty$ and the right boundary of the range of $X_t$ is  $\infty$.

\begin{lemma} \label{lem:finite}
Let $\alpha<\lambda$ and let $(\alpha,g)$ and $(\lambda,h)$ be candidate pairs.
Then,
\begin{equation*}
\begin{aligned}
\int_{-\infty}^{0}v^{-2}(y)\,dy\textnormal{ is finite if } h'(0)\geq g'(0)\;,\\
\int_{0}^{\infty}v^{-2}(y)\,dy\textnormal{ is finite if } h'(0)\leq g'(0)\;.
\end{aligned}
\end{equation*}
where $g=vq$ and $q(x)=e^{-\int_{0}^{x}\frac{k(y)}{\sigma^{2}(y)}\,dy}.$ 
\end{lemma}

\begin{proof} 
Write $h=uq.$ Assume that $h'(0)\geq g'(0),$ equivalently $u'(0)\geq v'(0).$
Define $\Gamma=\frac{u'}{u}-\frac{v'}{v}.$
Then
$$\Gamma'=-\Gamma^{2}-\frac{2v'}{v}\Gamma-\frac{2(\lambda-\alpha)}{\sigma^{2}}\;.$$
Because $\Gamma(0)\geq0,$ we have that $\Gamma(x)>0$ for $x<0,$
since if $\Gamma$ ever gets close to $0,$ then term $-\frac{2(\lambda-\alpha)}{\sigma^{2}}$ dominates the right hand side of the equation.
Choose $x_{0}$ with $x_{0}<0.$
For $x<x_{0},$ we have
$$-\frac{2v'(x)}{v(x)}=\frac{\Gamma'(x)}{\Gamma(x)}+\Gamma(x)+\frac{2(\lambda-\alpha)}
{\sigma^{2}(x)}\cdot\frac{1}{\Gamma(x)}\;.$$
Integrating from $x_{0}$ to $x,$
$$-2\ln\frac{v(\,x\,)}{v(x_{0})}=\ln \frac{\Gamma(\,x\,)}{\Gamma(x_{0})}+\int_{x_{0}}^{x}
\Gamma(y)\,dy+\int_{x_{0}}^{x}\frac{2(\lambda-\alpha)}{\sigma^{2}(y)}\cdot\frac{1}{\Gamma(y)}
\,dy$$
which leads to
$$\frac{v^{2}(x_{0})}{v^{2}(\,x\,)}\leq \frac{\Gamma(\,x\,)}{\Gamma(x_{0})}\,e^{\int_{x_{0}}^{x}
\Gamma(y)\,dy}$$
for $x<x_{0}.$
Thus,
\begin{equation*}
\begin{aligned}
\int_{-\infty}^{x_{0}}\frac{1}{v^{2}(y)}\,dy
&\leq (\text{constant})\cdot\int_{-\infty}^{x_{0}}\Gamma(y)\,e^{\int_{x_{0}}^{y}\Gamma(w)dw}\,dy 
\\
&=(\text{constant})\cdot\left(1-e^{-\int_{-\infty}^{x_{0}}\Gamma(w)dw}\right)\\
&\leq (\text{constant})\\
&< \infty \;.
\end{aligned}
\end{equation*}
This implies that $\int_{-\infty}^{0}v^{-2}(y)\,dy$ is finite.
Similarly, we can show that
$\int_{0}^{\infty}v^{-2}(y)\,dy$ is finite if $h'(0)\leq g'(0).$
\end{proof}

\begin{lemma} \label{lem:infinite}
 Write $H_\lambda=Vq$ and $h_\lambda=vq$, where $q(x):=e^{-\int_{0}^{x}\frac{k(y)}{\sigma^{2}(y)}\,dy}.$ Then,
$$\int_{-\infty}^{0}V^{-2}(y)\,dy=\infty\quad\text{ and }\quad \int_{0}^{\infty}v^{-2}(y)\,dy=\infty\;.$$ 
\end{lemma}
\begin{proof}
We only prove that $\int_{-\infty}^{0}V^{-2}(y)\,dy=\infty.$
A general (normalized to $h(0;c)=1$) solution of $\mathcal{L}h=-\lambda h$ is expressed by
$$h(x;c):=H_\lambda(x)\left(1+c\cdot \int_{0}^{x}V^{-2}(y)\,dy\right)\;,$$ 
thus we have
$$h'(0;c)=H_\lambda'(0)+c.$$
Since $H_\lambda'(0)$ is the maximum value by definition, we have that $c\leq 0.$
Suppose that $\int_{-\infty}^{0}V^{-2}(y)\,dy<\infty$ then we can choose a small positive number $c$ such that 
$h(x;c)$ is a positive function. This is a contradiction.
\end{proof}

\noindent We now prove Proposition \ref{prop:relation}.
\begin{proof}
We only show the inequality for $x<0.$ 
First, we prove $(h_\alpha'h_\alpha^{-1})(x)<(h'h^{-1})(x)$  for $x<0.$
Let $q(x):=e^{-\int_{0}^{x}\frac{k(y)}{\sigma^{2}(y)}\,dy}.$  
Write $h=uq$ and $h_\alpha=vq.$
Then we have $h'(0)>h_\alpha'(0),$ equivalently $u'(0)>v'(0).$ 
It is because, if not, by Lemma \ref{lem:finite}, we have 
$$\int_{0}^{\infty}v^{-2}(y)\,dy<\infty\;,$$ 
and this contradicts to Lemma \ref{lem:infinite}.
Now we define $\gamma=\frac{u'}{u}-\frac{v'}{v}.$
Then
$$\gamma'=-\gamma^{2}-\frac{2v'}{v}\gamma-\frac{2(\lambda-\alpha)}{\sigma^{2}}\;.$$
Because $\gamma(0)>0,$ we have that $\gamma(x)>0$ for $x<0,$
since if $\gamma$ ever gets close to $0,$ then term $-\frac{2(\lambda-\alpha)}{\sigma^{2}}$ dominates the right hand side of the equation.
Thus for $x<0,$ it is obtained that $(v'v^{-1})(x)<(u'u^{-1})(x),$ which implies that $(h_\alpha'h_\alpha^{-1})(x)<(h'h^{-1})(x).$ 

We now prove $(h'h^{-1})(x)<(H_\alpha'H^{-1}_\alpha)(x)$  
Write $H_\lambda=Vq.$
Then we have $h'(0)<H_\alpha'(0),$ equivalently $u'(0)<V'(0).$ 
Define $\Gamma:=\frac{u'}{u}-\frac{V'}{V}.$
Then
$$\Gamma'=-\Gamma^{2}-\frac{2V'}{V}\Gamma-\frac{2(\lambda-\alpha)}{\sigma^{2}}\;.$$
We claim that $\Gamma(x)<0$ for $x<0.$
Suppose there exists $x_0<0$ such that $\Gamma(x_0)\geq0.$
Then for all $z<x_0,$ it is obtained that $\Gamma(z)>0$
since if $\Gamma$ ever gets close to $0,$ then term $-\frac{2(\lambda-\alpha)}{\sigma^{2}}$ dominates the right hand side of the equation. 
For $z<x_0,$ we have
$$-\frac{2V'(z)}{V(z)}=\frac{\Gamma'(z)}{\Gamma(z)}+\Gamma(z)+\frac{2(\beta-\lambda)}
{\sigma^{2}(z)}\cdot\frac{1}{\Gamma(z)}\;.$$
Integrating from $x_{0}$ to $y,$
$$-2\ln\frac{V(\,y\,)}{V(x_{0})}=\ln \frac{\Gamma(\,y\,)}{\Gamma(x_{0})}+\int_{x_{0}}^{y}
\Gamma(z)\,dz+\int_{x_{0}}^{y}\frac{2(\lambda-\alpha)}{\sigma^{2}(z)}\cdot\frac{1}{\Gamma(z)}
\,dz$$
which leads to
$$\frac{V^{2}(x_{0})}{V^{2}(\,y\,)}\leq \frac{\Gamma(\,y\,)}{\Gamma(x_{0})}\,e^{\int_{x_{0}}^{y}
\Gamma(z)\,dz}$$
for $y<x_{0}.$
Thus,
\begin{equation*}
\begin{aligned}
\int_{-\infty}^{x_{0}}V^{-2}(y)\,dy
&\leq (\text{constant})\cdot\int_{-\infty}^{x_{0}}\Gamma(y)\,e^{\int_{x_{0}}^{y}\Gamma(w)dw}\,dy 
\\
&=(\text{constant})\cdot\left(1-e^{-\int_{-\infty}^{x_{0}}\Gamma(w)dw}\right)\\
&\leq (\text{constant})\\
&< \infty \;.
\end{aligned}
\end{equation*}
This implies that $\int_{-\infty}^{0}V^{-2}(y)\,dy$ is finite.
This contradicts to Lemma \ref{lem:infinite}.
\end{proof}

\end{document}